\documentclass[11pt]{amsart}
\def\isdraft{0}

\usepackage{amsaddr} 
\usepackage{mathtools}
\usepackage{amsmath,amssymb,amsfonts,dsfont}
\usepackage{prettyref} 
\usepackage[utf8]{inputenc}
\usepackage{enumitem}
\usepackage[hyphens]{url}
\usepackage{tikz-cd}
\usepackage{tikz}
\usetikzlibrary{positioning}
\usetikzlibrary{arrows}
\usepackage{bussproofs}
\usepackage[mathscr]{euscript}
\usepackage{hyperref} 
\usepackage{amsthm}
\usepackage{theapa}
\usepackage{a4wide}
\usepackage[color=black,textcolor=white\if\isdraft0,disable\fi]{todonotes}
\usepackage{stackengine}
\usepackage[normalem]{ulem}
\usepackage{stmaryrd}
\usepackage{graphicx,txfonts}
\usepackage{wrapfig}
\usepackage{pdfpages}
\usepackage{float}
\graphicspath{{fig/}}
\usepackage{bm}

\newtheorem{theorem}{Theorem}

\newtheorem{lemma}[theorem]{Lemma}
\newtheorem{fact}[theorem]{Fact}

\theoremstyle{definition} 

\newtheorem{convention}[theorem]{Convention}

\newtheorem{remark}[theorem]{Remark}
\newtheorem{example}[theorem]{Example}

\newtheorem{problem}{Problem}

\newrefformat{cha}{Chapter \ref{#1}}
\newrefformat{sec}{Section \ref{#1}}
\newrefformat{tab}{Table \ref{#1}}
\newrefformat{fig}{Figure \ref{#1}}
\newrefformat{equ}{(\ref{#1})}
\newrefformat{app}{Appendix \ref{#1}}
\newrefformat{thm}{Theorem \ref{#1}}
\newrefformat{cor}{Corollary \ref{#1}}
\newrefformat{prop}{Proposition \ref{#1}}
\newrefformat{lem}{Lemma \ref{#1}}
\newrefformat{fact}{Fact \ref{#1}}
\newrefformat{obs}{Observation \ref{#1}}
\newrefformat{note}{Note \ref{#1}}
\newrefformat{idea}{Idea \ref{#1}}
\newrefformat{trivia}{Trivia \ref{#1}}
\newrefformat{def}{Definition \ref{#1}}
\newrefformat{not}{Notation \ref{#1}}
\newrefformat{con}{Convention \ref{#1}}
\newrefformat{rem}{Remark \ref{#1}}
\newrefformat{exa}{Example \ref{#1}}
\newrefformat{problem}{Problem \ref{#1}}
\newrefformat{claim}{Claim \ref{#1}}
\newrefformat{conjecture}{Conjecture \ref{#1}}
\newrefformat{exe}{Exercise \ref{#1}}
\newrefformat{alg}{Algorithm \ref{#1}}
\newrefformat{err}{Error \ref{#1}}
\newrefformat{que}{Question \ref{#1}}
\newrefformat{ite}{Item \ref{#1}}
\newrefformat{Q}{Q. \ref{#1}}
\newrefformat{Todo}{Todo \ref{#1}}
\newrefformat{warning}{Warning \ref{#1}}
\newrefformat{pseudocode}{Pseudocode \ref{#1}}

\newcommand{\righttherefore}{:\joinrel\cdot\,}

\title{Bilingual analogical proportions via hedges}
\author{
	Christian Anti\'c
}
\address{
	christian.antic@icloud.com\\
	Vienna University of Technology\\
	Vienna, Austria
}

\begin{document}
\begin{abstract}
	Analogical proportions are expressions of the form ``$a$ is to $b$ what $c$ is to $d$'' at the core of analogical reasoning which itself is at the core of human and artificial intelligence. The author has recently introduced {\em from first principles} an abstract algebro-logical framework of analogical proportions within the general setting of universal algebra and first-order logic. In that framework, the source and target algebras have the {\em same} underlying language. The purpose of this paper is to generalize his unilingual framework to a bilingual one where the underlying languages may differ. This is achieved by using hedges in justifications of proportions. The outcome is a major generalization vastly extending the applicability of the underlying framework. In a broader sense, this paper is a further step towards a mathematical theory of analogical reasoning.
\end{abstract}
\maketitle

\section{Introduction and preliminaries}

Analogical proportions are expressions of the form ``$a$ is to $b$ what $c$ is to $d$'' at the core of analogical reasoning which itself is at the core of human and artificial intelligence with applications to such diverse tasks as proving mathematical theorems and building mathematical theories, commonsense reasoning, learning, language acquisition, and story telling \shortcite<e.g.>{Barbot19,Boden98,Gust08,Hofstadter01,Hofstadter13,Krieger03,Miclet09,Polya54,Prade18,Prade21,Schmidt14,Stroppa06,Winston80}. 

\citeA{Antic22,Antic22-3,Antic23-4} has recently introduced {\em from first principles} an abstract algebro-logical framework of analogical proportions within the general setting of universal algebra and first-order logic. It is a promising novel model of analogical proportions with appealing mathematical properties \cite<cf.>{Antic21-3,Antic23-7}. 

Anti\'c's framework is unilingual in the sense that the underlying languages (or similarity types) of the source and target algebras are identical. This means that for example analogical proportions over the source algebra $(\mathbb N,S)$ --- here $S$ is the unary successor function --- and the target algebra $(\mathbb N,+)$ cannot be formulated since $S$ and $+$ have different ranks. In this paper, we shall generalize the framework by allowing the underlying languages to differ (see \prettyref{exa:1_2__2_4}). This vastly extends the applicability of the framework. Technically, we replace terms in justifications by (restricted) hedges \cite<cf.>{Yamamoto01,Kutsia14}.

In \prettyref{sec:Properties}, we obtain bilingual generalizations of the Uniqueness Lemma (see \prettyref{lem:BUL}) and the Functional Proportion Theorem (see \prettyref{thm:BFPT}). Moreover, in \prettyref{thm:p-axioms} we observe that the generalized bilingual framework of this paper preserves all desired properties listed in \citeA[§4.3]{Antic22} which is further evidence for the robustness of the underlying framework. 

Finally, \prettyref{sec:Problems} lists some problems which remained unsolved in this paper and appear to be interesting lines of future research.

In a broader sense, this paper is a further step towards a mathematical theory of analogical reasoning.

\subsection*{Preliminaries}

We assume the reader to be fluent in basic universal algebra as it is presented for example in \citeA[§II]{Burris00}. 

A {\em language} $\mathcal L$ of algebras consists of a set of ranked {\em function symbols} together with an associated {\em rank function} $r_\mathcal L:\mathcal L\to\mathbb N$ (often simply denoted by $r$), and a denumerable set $\mathcal Z=\{z,z_1,z_2,\ldots\}$ of {\em $\mathcal Z$-variables}. The sets $\mathcal{L,Z}$ are pairwise disjoint. We denote the set of function symbols of rank $n$ by $\mathcal L_n$ so that $\mathcal L_0$ denotes the {\em constant symbols}. The set $\mathcal{T_L(X,Z)}$ of {\em $\mathcal L$-terms} with variables among $\mathcal{X\cup Z}$ is defined as usual. We denote the set of $\mathcal Z$-variables occurring in a term $s$ by $\mathcal Z(s)$. We say that a term $s$ has {\em rank} $n$ iff $\mathcal Z(s)=\{z_1,\ldots,z_n\}$. A term is {\em ground} iff it contains no constant symbols. Terms can be interpreted as ``generalized elements'' containing variables as placeholders for concrete elements.

\begin{convention}[Dot convention] We add a dot to formal symbols to distinguish them from the objects that they are inteded to denote --- for example, $\dot f$ is a function symbol which stands for a function $f$.
\end{convention}

An {\em $\mathcal L$-algebra} $\mathfrak A$ consists of:
\begin{itemize}
	\item a non-empty set $A$, the {\em universe} of $\mathfrak A$;
	\item for each $\dot f\in\mathcal L$ a function $\dot f^\mathfrak A:A^{r_\mathcal L(\dot f)}\to A$, the {\em functions} of $\mathfrak A$;
	\item for each constant symbol $\dot c\in\mathcal L_0$, an element $\dot c^\mathfrak A\in A$, the {\em distinguished elements} of $\mathfrak A$.
\end{itemize}

An {\em $\mathfrak A$-assignment} is a function $\alpha:\mathcal Z\to A$ mapping $\mathcal Z$-variables to elements of $A$. 
In what follows, we always write the application of an assignment in postfix notation. Relative to an $\mathfrak A$-assignment $\alpha$, we define the denotation $s^\mathfrak A\alpha$ of an $\mathcal L$-term $s$ in $\mathfrak A$ with respect to $\alpha$ inductively as follows:
\begin{itemize}
	\item for a variable $z\in\mathcal Z$, $z^\mathfrak A\alpha:=\alpha(z)$;
	\item for a constant symbol $\dot c\in\mathcal L_0$, $\dot c^\mathfrak A\alpha:=c^\mathfrak A$;
	\item for a function symbol $\dot f\in\mathcal L$ and $\mathcal L$-terms $s_1,\ldots,s_n\in\mathcal{T_L(Z)}$,
	\begin{align*} 
		(\dot f(s_1,\ldots,s_{r(\dot f)}))^\mathfrak A\alpha:= f^\mathfrak A(s_1^\mathfrak A\alpha,\ldots s_{r(\dot f)}^\mathfrak A\alpha).
	\end{align*}
\end{itemize} Notice that $s^\mathfrak A$ induces a function, which we again denote by $s^\mathfrak A$, given by
\begin{align*} 
	s^\mathfrak A(z_1,\ldots,z_{r(s)}):A^{r(s)}\to A:(a_1,\ldots,a_{r(s)})\mapsto s^\mathfrak A(a_1,\ldots,a_{r(s)}):=s^\mathfrak A\{z_1/a_1,\ldots,z_{r(s)}/a_{r(s)}\}.
\end{align*} We call a term $s$ {\em injective} in $\mathfrak A$ iff $s^\mathfrak A$ is an injective function. Moreover, we call $s$ {\em constant} in $\mathfrak A$ iff $s^\mathfrak A$ is a constant function --- notice that ground terms induce constant functions which can be identified with elements of $A$.

\section{Bilingual analogical proportions}

This is the main section of the paper. Here we shall generalize \citeS{Antic22,Antic22-3} abstract algebraic framework of analogical proportions from an unilingual to a bilingual setting where the underlying languages may differ. 

For this, let $\mathfrak A$ and $\mathfrak B$ be algebras over languages of algebras $\mathcal L_\mathfrak A$ and $\mathcal L_\mathfrak B$, respectively, where we require $|\mathcal L_\mathfrak A|=|\mathcal L_\mathfrak B|$. That is, we assume that both languages are indexed by the same index set $I$ such that
\begin{align*} 
	\mathcal L_\mathfrak A=\{\dot f_i\mid i\in I\} \quad\text{and}\quad \mathcal L_\mathfrak B=\{\dot g_i\mid i\in I\},
\end{align*} and thus
\begin{align*} 
	\mathfrak A=\left(A,\left\{\dot f_i^\mathfrak A \;\middle|\; i\in I\right\}\right) \quad\text{and}\quad \mathfrak B=\left(B,\left\{\dot g_i^\mathfrak B \;\middle|\; i\in I\right\}\right).
\end{align*} This is not a serious restriction since we can always add copies of a function --- for example, $(\mathbb N,S_1,S_2)$ contains two copies of the same successor function and has the same index set as $(\mathbb N,+,\cdot)$. 

Now let $\mathcal L$ be a language generalizing $\mathcal L_\mathfrak A$ and $\mathcal L_\mathfrak B$ in the sense that it contains a function symbol $\dot h_i$ for each $\dot f_i\in\mathcal L_\mathfrak A$ and $\dot g_i\in\mathcal L_\mathfrak B$, $i\in I$, where we define the rank function of $\mathcal L$ by
\begin{align*} 
	r_\mathcal L(\dot h_i):=\max\{r_{\mathcal L_\mathfrak A}(\dot f_i),r_{\mathcal L_\mathfrak B}(\dot g_i)\},\quad i\in I.
\end{align*} Moreover, we assume a set of {\em hedge variables}\footnote{See \citeA{Yamamoto01} and \citeA{Kutsia14}.} $\mathscr Z:=\{Z,Z_1,Z_2,\ldots\}$ which may be replaced by the {\em empty hedge} $\dot\lambda$ and which we will use to formally reduce the rank of a term --- for example, $\dot h(\dot a,\dot\lambda)$ will stand for $\dot h(\dot a)$. The use of hedge variables will be necessary in cases where function symbols of different ranks are to be generalized. In addition, we assume a set $\mathcal X=\{x,x_1,x_2,\ldots\}$ of {\em $\mathcal X$-variables} as placeholders for constant symbols \cite<see>{Antic22-3}. An {\em $\mathcal L$-hedge} \cite<cf.>{Yamamoto01,Kutsia14} is an $\mathcal L$-term possibly containing variables from
\begin{align*} 
	\mathcal{V:=X\cup Z}\cup\mathscr Z,
\end{align*} and we will use the notation $\tilde s\mathbf{[x](z,Z)}$ for $\mathcal L$-hedges. This situation can be depicted as follows:
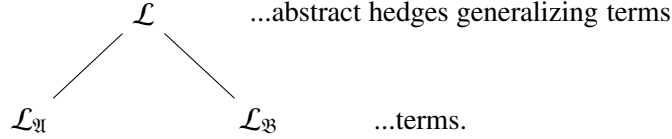
\begin{center}
\begin{tikzpicture} 
	\node (L) {$\mathcal L$};
	\node (label) [right=of L] {...abstract hedges generalizing terms};
	\node (L') [below=of L] {};
	\node (L_A) [left=of L'] {$\mathcal L_\mathfrak A$};
	\node (L_B) [right=of L'] {$\mathcal L_\mathfrak B$};
	\node (label2) [right=of L_B] {...terms.};
	\draw (L) to (L_A);
	\draw (L) to (L_B);
\end{tikzpicture}
\end{center} We denote the set of all $\mathcal L$-hedges with variables from $\mathcal V$ by $\mathcal{H_L(V)}$. An {\em abstract} $\mathcal L$-hedge contains no constant symbols. We will use the language $\mathcal L$ to express $\mathcal L$-justifications of analogical proportions below.

An {\em $(\mathcal{L,L_\mathfrak A})$-substitution} is a function
\begin{align*} 
	\sigma_\mathfrak A:\mathcal{L\cup H_L(V)}\to\mathcal{L_\mathfrak A\cup T_{L_\mathfrak A\cup\{\dot\lambda\}}(X,Z)}
\end{align*} with signatures
\begin{align*}
	&(\sigma_\mathfrak A\upharpoonright\mathcal L):\mathcal L\to\mathcal L_\mathfrak A:\dot h_i\mapsto\dot f_i,\quad i\in I,\\
	&(\sigma_\mathfrak A\upharpoonright\mathcal X):\mathcal X\to\mathcal L_{\mathfrak A,0},\\
	&(\sigma_\mathfrak A\upharpoonright\mathcal Z):\mathcal Z\to\mathcal{T_{L_\mathfrak A}(X,Z)},\\
	&(\sigma_\mathfrak A\upharpoonright\mathscr Z):\mathscr Z\to\mathcal{T_{L_\mathfrak A\cup\{\dot\lambda\}}(X,Z)}.
\end{align*} which is the identity almost everywhere, extended to hedges inductively in the usual way. Notice that $\sigma_\mathfrak A$ is {\em fixed} on $\mathcal L$ and maps each function (and constant) symbol $\dot h_i$ to $\dot f_i$, for every $i\in I$. For example, we will write $\{\dot h_i/\dot f_i,x/\dot a,z/\dot f(\dot a),Z/\dot\lambda\}$ for the substitution that is the identity on every variable except for $\dot h_i,x,z$, and $Z$ which are mapped respectively to $\dot f_i$, $\dot a$, $\dot f(\dot a)$, and $\dot\lambda$. A substitution is {\em ground} iff it has the signature
\begin{align*} 
	\mathcal{L\cup H_L(V)}\to\mathcal{L_\mathfrak A\cup T_{L_\mathfrak A\cup\{\dot\lambda\}}(\emptyset)}.
\end{align*} We denote the set of all ground $(\mathcal{L,L_\mathfrak A})$-substitutions by $g\text-Sub(\mathcal{L,L_\mathfrak A})$. The application of an $(\mathcal{L,L_\mathfrak A})$-substitution $\sigma_\mathfrak A$ to an $\mathcal L$-hedge $\tilde s$ is written in postfix notation $\tilde s\sigma_\mathfrak A$ and is defined inductively as usual. 


\begin{example}\label{exa:h(z,Z)} Preparatory to \prettyref{exa:1_2__2_4}, consider the algebras
\begin{align*} 
	\mathfrak N:=(\mathbb N,S,0) \quad\text{and}\quad \mathfrak M:=(\mathbb N,+,1),
\end{align*} where $S:\mathbb{N\to N}$ is the unary {\em successor function} $S(a):=a+1$, for every $a\in\mathbb N$, and let (here ``$\dot\mathunderscore$'' is a dummy constant symbol)\footnote{The algebras $\mathfrak N$ and $\mathfrak M$ will be defined in \prettyref{exa:1_2__2_4}.}
\begin{align*} 
	\mathcal L_\mathfrak N:=\{\dot S,\dot 0\} \quad\text{and}\quad \mathcal L_\mathfrak M:=\{\dot +,\dot 1\} \quad\text{and}\quad \mathcal L:=\{\dot h,\dot\mathunderscore\}
\end{align*} with ranks
\begin{align*} 
	r(\dot S)&:=1\\
	r(\dot 0)&:=0\\
	r(\dot +)&:=2\\
	r(\dot 1)&:=0,\\
	r(\dot h)&\;=\max\{r(\dot S),r(\dot +)\}=2\\
	r(\dot\mathunderscore)&\;=\max\{r(\dot 0),r(\dot 1)\}=0.
\end{align*} Moreover, let $\mathcal X:=\emptyset$, $\mathcal Z:=\{z\}$, and $\mathscr Z:=\{Z\}$. Then the hedge $h(z,Z)$ generalizes both ground terms $\dot S(\dot 0)$ and $\dot 1\;\dot +\;\dot 1$ via
\begin{align*} 
	\sigma_\mathfrak N:=\{\dot h/\dot S,z/\dot 0,Z/\dot\lambda\} \quad\text{and}\quad \sigma_\mathfrak M:=\{\dot h/\dot +,z/\dot 1,Z/\dot 1\},
\end{align*} where $\dot +$ is written in infix notation for readability, and $\dot h(\dot h(z,Z),Z)$ generalizes $\dot S(\dot S(\dot 0))$ and $(\dot 1\;\dot +\;\dot 1)\;\dot +\;\dot 1$ again via $\sigma_\mathfrak N$ and $\sigma_\mathfrak M$. Formally, we obtain
\begin{align*} 
	 h(z,Z)\sigma_\mathfrak N=\dot S(\dot 0,\dot\lambda)=\dot S(\dot 0) \quad\text{and}\quad  h(z,Z)\sigma_\mathfrak M=\dot +(\dot 1,\dot 1)
\end{align*} and
\begin{align*} 
	 h(\dot h(z,Z),Z)\sigma_\mathfrak N=\dot S(\dot S(\dot 0,\dot\lambda),\dot\lambda)=\dot S(\dot S(\dot 0)) \quad\text{and}\quad \dot h(\dot h(z,Z),Z)\sigma_\mathfrak B=\dot +(\dot +(\dot 1,\dot 1),\dot 1).
\end{align*}
\end{example}

\begin{convention}\label{con:unilingual} We make the convention that in case $\mathcal L_\mathfrak A=\mathcal L_\mathfrak B$ and $r_{\mathcal L_\mathfrak A}=r_{\mathcal L_\mathfrak B}$ then $\mathcal L=\mathcal L_\mathfrak A=\mathcal L_\mathfrak B$ turning $\mathcal L$ into a ranked language with $r_\mathcal L=r_{\mathcal L_\mathfrak A}=r_{\mathcal L_\mathfrak B}$. This guarantees that our bilingual framework of this paper coincides with the unilingual one in \citeA{Antic22,Antic22-3} in case the underlying languages are identical.
\end{convention}


An {\em $\mathcal L$-justification} with variables among $\mathcal V$ is an expression of the form
\begin{align*} 
	\tilde s\mathbf{[x](z,Z)}\to\tilde t\mathbf{[x](z,Z)}
\end{align*} consisting of two abstract $\mathcal L$-hedges $\tilde s$ and $\tilde t$ with $\mathcal X$-variables $\mathbf x$, $\mathcal Z$-variables $\mathbf z$, hedge $\mathscr Z$-variables $\mathbf Z$, and an arrow $\to$ which by convention binds weaker than any other function symbol, where we require that every $\mathcal Z$-variable in $\tilde t$ occurs in $\tilde s$, that is, $\mathcal Z(\tilde t)\subseteq\mathcal Z(\tilde s)$. We denote the set of all such $\mathcal L$-justifications by $\mathcal{J_L(V)}$.

We are now ready to introduce the main notion of the paper by following the lines of \citeA{Antic22,Antic22-3}. Let $a,b\in\mathcal{T_{L_\mathfrak A}(\emptyset)}$ and $c,d\in\mathcal{T_{L_\mathfrak A}(\emptyset)}$ be ground terms standing for elements in $\mathfrak A$ and $\mathfrak B$, respectively. We define the {\em analogical proportion relation} in two steps:
\begin{enumerate}
	\item Define the {\em set of $\mathcal L$-justifications} of an {\em arrow} $a\to b$ by
	\begin{align*} 
		Jus_\mathfrak A(a\to b)
			:= \left\{\tilde s\mathbf{[x](z,Z)}\to\tilde t\mathbf{[x](z,Z)}\in \mathcal{J_L(V)} 
			\;\middle|\; 
			\begin{array}{c}
				a^\mathfrak A\to b^\mathfrak A=(\tilde s\sigma_\mathfrak A)^\mathfrak A\to (\tilde t\sigma_\mathfrak A)^\mathfrak A\\
				\sigma_\mathfrak A\in g\text-Sub(\mathcal{L,L_\mathfrak A})
			\end{array}
			\right\}
	\end{align*} extended to an {\em arrow proportion} $a\to b\righttherefore c\to d$ by
	\begin{align*} 
		Jus_{(\mathfrak{A,B})}&(a\to b\righttherefore c\to d)\\
			&:= 
			\left\{
				\begin{array}{c}
					\tilde s\mathbf{[x](z,Z)}\to \tilde t\mathbf{[x](z,Z)}\\
					\in\mathcal{J_L(V)}
				\end{array}
				\;\middle|\;
				\begin{array}{c}
					a^\mathfrak A\to b^\mathfrak A=(\tilde s\sigma_\mathfrak A)^\mathfrak A\to (\tilde t\sigma_\mathfrak A)^\mathfrak A\\
					c^\mathfrak A\to d^\mathfrak A=(\tilde s\sigma_\mathfrak B)^\mathfrak B\to (\tilde t\sigma_\mathfrak B)^\mathfrak B\\
					\sigma_\mathfrak A\in g\text-Sub(\mathcal{L,L_\mathfrak A})\\
					\sigma_\mathfrak B\in g\text-Sub(\mathcal{L,L_\mathfrak B})\\
					(\sigma_\mathfrak A\upharpoonright\mathcal X)=(\sigma_\mathfrak B\upharpoonright\mathcal X)
				\end{array}
			\right\}.
	\end{align*} We call an $\mathcal L$-justification {\em trivial in $\mathfrak A$} iff it is contained in {\em every} $Jus_\mathfrak A(a\to b)$, for {\em all} $a,b\in A$; moreover, we call an $\mathcal L$-justification {\em trivial in $(\mathfrak{A,B})$} iff it is trivial in $\mathfrak A$ and $\mathfrak B$,\footnote{For example, $z\to x$ and $x_1\to x_2$ are trivial in every pair of algebras.} and we say that $J$ is a {\em trivial set of $\mathcal L$-justifications} in $(\mathfrak{A,B})$ iff every $\mathcal L$-justification in $J$ is trivial in $(\mathfrak{A,B})$. We say that $a\to b \righttherefore c\to d$ {\em holds} in $(\mathfrak{A,B})$ --- in symbols,
	\begin{align*} 
		(\mathfrak{A,B})\models a\to b \righttherefore c\to d,
	\end{align*} iff
	\begin{enumerate}
		\item either $Jus_\mathfrak A(a\to b)\cup Jus_\mathfrak B(c\to d)$ consists only of trivial $\mathcal L$-justifications, in which case there is neither a non-trivial transformation of $a$ into $b$ in $\mathfrak A$ nor of $c$ into $d$ in $\mathfrak B$; or

		\item $Jus_{(\mathfrak{A,B})}(a\to b\righttherefore c\to d)$ is maximal with respect to subset inclusion among the sets $Jus_{(\mathfrak{A,B})}(a\to b\righttherefore c\to d')$, $d'\in\mathcal{T_{L_\mathfrak A}(\emptyset)}$ with ${d'}^\mathfrak B\neq d^\mathfrak B$, containing at least one non-trivial $\mathcal L$-justification, that is, for any such ground term $d'$,\footnote{In what follows, we will usually omit trivial $\mathcal L$-justifications from notation. So, for example, we will write $Jus_{(\mathfrak{A,B})}(a\to b\righttherefore c\to d)=\emptyset$ instead of $Jus_{(\mathfrak{A,B})}(a\to b\righttherefore c\to d)=\{\text{trivial $\mathcal L$-justifications}\}$ in case $a\to b\righttherefore c\to d$ has only trivial $\mathcal L$-justifications in $(\mathfrak{A,B})$, et cetera. The empty set is always a trivial set of $\mathcal L$-justifications. Every $\mathcal L$-justification is meant to be non-trivial unless stated otherwise. Moreover, we will always write sets of $\mathcal L$-justifications modulo renaming of variables, that is, we will write $\{z\to z\}$ instead of $\{z\to z\mid z\in Z\}$ et cetera.}
		\begin{align*} 
		    \emptyset\subsetneq Jus_{(\mathfrak{A,B})}(a\to b\righttherefore c\to d)&\subseteq Jus_{(\mathfrak{A,B})}(a\to b\righttherefore c\to d')
		\end{align*} implies
		\begin{align*} 
		    \emptyset\subsetneq Jus_{(\mathfrak{A,B})}(a\to b\righttherefore c\to d')\subseteq Jus_{(\mathfrak{A,B})}(a\to b\righttherefore c\to d).
		\end{align*}
	\end{enumerate} We abbreviate the above definition by simply saying that $Jus_{(\mathfrak{A,B})}(a\to b\righttherefore c\to d)$ is {\em d-maximal}.

	\item Finally, the analogical proportion relation is most succinctly defined by the following derivation:
    \begin{prooftree}
    	\AxiomC{$a\to b\righttherefore c\to d\qquad b\to a\righttherefore d\to c$}
    	\AxiomC{$c\to d\righttherefore a\to b\qquad d\to c\righttherefore b\to a$}
    	\BinaryInfC{$a:b::c:d$.}
    \end{prooftree} This means that in order to prove $(\mathfrak{A,B})\models a:b::c:d$, we need to check the four relations in the first line with respect to $\models$ in $(\mathfrak{A,B})$. The set of all analogical proportions which hold in $(\mathfrak{A,B})$ is denoted by $AP(\mathfrak{A,B})$, that is,
    \begin{align*} 
		AP(\mathfrak{A,B}):=\left\{a:b::c:d \;\middle|\; (\mathfrak{A,B})\models a:b::c:d\right\}.
	\end{align*}
\end{enumerate} We will always write $\mathfrak A$ instead of $(\mathfrak{A,A})$.

\begin{example}\label{exa:1_2__2_4} This example is a continuation of \prettyref{exa:h(z,Z)}. The arrow proportion
\begin{align*} 
	\dot S(\dot 0)\to\dot S(\dot S(\dot 0)) \righttherefore \dot 1\;\dot +\;\dot 1\to \dot 1\;\dot +\;\dot 1\;\dot +\;\dot 1
\end{align*} stands for
\begin{align*} 
	1\to 2 \righttherefore 2\to 3
\end{align*} and has the justification
\begin{align*} 
	\dot h(z,Z)\to\dot h(\dot h(z,Z),Z)
\end{align*} in $(\mathfrak{N,M})$ since
\begin{align*} 
	1\to 2&=(\dot h(z,Z)\{\dot h/\dot S,z/\dot 0,Z/\dot\lambda\})^\mathfrak N\to (\dot h(\dot h(z,Z),Z)\{\dot h/\dot S,z/\dot 0,Z/\dot\lambda\})^\mathfrak N,\\
	2\to 3&=(\dot h(z,Z)\{\dot h/\dot +,z/\dot 1,Z/\dot 1\})^\mathfrak M\to (\dot h(\dot h(z,Z),Z)\{\dot h/\dot +,z/\dot 1,Z/\dot 1\})^\mathfrak M.
\end{align*} This can be depicted as follows:
\begin{center}
\begin{tikzpicture}[node distance=1cm and 0.5cm]
\node (a)               {$\dot S(\dot 0)$};
\node (d1) [right=of a] {$\to$};
\node (b) [right=of d1] {$\dot S(\dot S(\dot 0))$};
\node (d2) [right=of b] {$\righttherefore$};
\node (c) [right=of d2] {$\dot 1\;\dot +\;\dot 1$};
\node (d3) [right=of c] {$\to$};
\node (d) [right=of d3] {$(\dot 1\;\dot +\;\dot 1)\;\dot +\;\dot 1$.};
\node (s) [below=of b] {$\dot h(z,Z)$};
\node (t) [above=of c] {$\dot h(\dot h(z,Z),Z)$};

\draw (a) to [edge label'={$\{\dot h/\dot S,z/\dot 0,Z/\dot\lambda\}$}] (s); 
\draw (c) to [edge label={$\{\dot h/\dot +,z/\dot 1,Z/\dot 1\}$}] (s);
\draw (b) to [edge label={$\{\dot h/\dot S,z/\dot 0,Z/\dot\lambda\}$}] (t);
\draw (d) to [edge label'={$\{\dot h/\dot +,z/\dot 1,Z/\dot 1\}$}] (t);
\end{tikzpicture}
\end{center} In \prettyref{exa:char}, we will see that this justification is in fact a ``characteristic'' one thus yielding
\begin{align*} 
	(\mathfrak{N,M})\models \dot S(\dot 0)\to\dot S(\dot S(\dot 0)) \righttherefore \dot 1\;\dot +\;\dot 1\to \dot 1\;\dot +\;\dot 1\;\dot +\;\dot 1.
\end{align*} Of course, this arrow proportion cannot be entailed in the unilingual framework since $S$ and $+$ have different ranks!
\end{example}

\section{Properties}\label{sec:Properties}

Recall that the {\em clone} of an $\mathcal L$-algebra $\mathfrak A$ is given by
\begin{align*} 
	Clo(\mathfrak A):= \left\{s^\mathfrak A \;\middle|\; s\in\mathcal{T_L(Z)}\right\}.
\end{align*} The following observation follows immediately from definitions,\todo{check} and it says that the framework is invariant under identical clones:

\begin{fact} We have the following implication:
\begin{prooftree}
	\AxiomC{$Clo(\mathfrak A)=Clo(\mathfrak A')$}
	\AxiomC{$Clo(\mathfrak B)=Clo(\mathfrak B')$}
	\BinaryInfC{$AP(\mathfrak{A,B})=AP(\mathfrak{A',B'})$.}
\end{prooftree}
\end{fact}

Moreover, the framework is invariant with respect to semantic equivalence of ground terms:

\begin{fact}\label{fact:ground} Given ground terms $a,b,a',b'\in\mathcal{T_{L_\mathfrak A}(\emptyset)}$ and $c,d,c',d'\in\mathcal{T_{L_\mathfrak B}(\emptyset)}$ satisfying
\begin{align*} 
	a^\mathfrak A={a'}^\mathfrak A \quad\text{and}\quad b^\mathfrak A={b'}^\mathfrak A \quad\text{and}\quad c^\mathfrak B={c'}^\mathfrak B \quad\text{and}\quad d^\mathfrak B={d'}^\mathfrak B,
\end{align*} then
\begin{align*} 
	(\mathfrak{A,B})\models a:b::c:d \quad\Leftrightarrow\quad (\mathfrak{A,B})\models a':b'::c':d'.
\end{align*}
\end{fact}

\subsection*{Characteristic justifications}

Computing all $\mathcal L$-justifications of an arrow proportion is difficult in general, which fortunately can be omitted in many cases. The following definition is essentially the same as \citeS[Definition 20]{Antic22}. We call a set $J$ of $\mathcal L$-justifications a {\em characteristic set of $\mathcal L$-justifications} of $a\to b\righttherefore c\to d$ in $(\mathfrak{A,B})$ iff $J$ is a sufficient set of $\mathcal L$-justifications of $a\to b\righttherefore c\to d$ in $(\mathfrak{A,B})$, that is, iff
\begin{enumerate}
	\item $J\subseteq Jus_{(\mathfrak{A,B})}(a\to b\righttherefore c\to d)$, and
	\item $J\subseteq Jus_{(\mathfrak{A,B})}(a\to b\righttherefore c\to d')$ implies $d'=d$, for each $d'\in\mathcal{T_{L_\mathfrak B}(\emptyset)}$.
\end{enumerate} In case $J=\{\tilde s\to\tilde t\}$ is a singleton set satisfying both conditions, we call $\tilde s\to\tilde t$ a {\em characteristic $\mathcal L$-justification} of $a\to b\righttherefore c\to d$ in $(\mathfrak{A,B})$.

\begin{example}\label{exa:char} Recall the situation in \prettyref{exa:1_2__2_4}. Since $\sigma_\mathfrak M=\{\dot h/\dot +,z/\dot 1,Z/\dot 1\}$ is the {\em unique} substitution yielding
\begin{align*} 
	\dot h(z,Z)\sigma_\mathfrak M=\dot 1\;\dot +\;\dot 1,
\end{align*} we deduce that
\begin{align*} 
	\dot h(z,Z)\to\dot h(\dot h(z,Z),Z)
\end{align*} is indeed a {\em characteristic} justification of
\begin{align*} 
	\dot S(\dot 0)\to\dot S(\dot S(\dot 0)) \righttherefore \dot 1\;\dot +\;\dot 1\to \dot 1\;\dot +\;\dot 1\;\dot +\;\dot 1
\end{align*} in $(\mathfrak{N,M})$.
\end{example}

The following result provides a sufficient condition of characteristic $\mathcal L$-justifications. It is a generalization of the Extended Uniqueness Lemma in \citeA{Antic22-3} to the bilingual setting --- it is almost identical to the original characterization which demonstrates the robustness of the underlying framework!

\begin{lemma}[Bilingual Uniqueness Lemma]\label{lem:BUL} Let $\tilde s\mathbf{[x](z,Z)}\to \tilde t\mathbf{[x](z,Z)}$ be a non-trivial $\mathcal L$-justification of $a\to b\righttherefore c\to d$ in $(\mathfrak{A,B})$.
\begin{enumerate}
	\item If
	\begin{align}\label{equ:UL1} 
		c^\mathfrak B=(\tilde s\sigma_\mathfrak B)^\mathfrak B \quad\text{implies}\quad d^\mathfrak B=(\tilde t\sigma_\mathfrak B)^\mathfrak B,
	\end{align} for all $\sigma_\mathfrak B\in g\text-Sub(\mathcal{L,L_\mathfrak B})$, then $\tilde s\to \tilde t$ is a {\em characteristic} $\mathcal L$-justification of $a\to b\righttherefore c\to d$ in $(\mathfrak{A,B})$, that is,
	\begin{align*} 
		(\mathfrak{A,B})\models a\to b\righttherefore c\to d.
	\end{align*}

	\item Consequently, if
	\begin{align*} 
		c^\mathfrak B=(\tilde s\sigma_\mathfrak B)^\mathfrak B \quad\Leftrightarrow\quad d^\mathfrak B=(\tilde s\sigma_\mathfrak B)^\mathfrak B,
	\end{align*} for all $\sigma_\mathfrak B\in g\text-Sub(\mathcal{L,L_\mathfrak B})$, and every $\mathcal Z$-variable in $\tilde s$ occurs in $\tilde t$ (recall that every $\mathcal Z$-variable in $\tilde t$ occurs in $\tilde s$ by assumption), then
	\begin{align*} 
		(\mathfrak{A,B})\models a\to b\righttherefore c\to d \quad\text{and}\quad (\mathfrak{A,B})\models b\to a\righttherefore d\to c.
	\end{align*}
	\item Consequently, if
	\begin{align*}
		&a^\mathfrak A=(\tilde s\sigma_\mathfrak A)^\mathfrak A \quad\Leftrightarrow\quad b^\mathfrak A=(\tilde t\sigma_\mathfrak A)^\mathfrak B,\\
		&c^\mathfrak B=(\tilde s\sigma_\mathfrak B)^\mathfrak B \quad\Leftrightarrow\quad d^\mathfrak B=(\tilde t\sigma_\mathfrak B)^\mathfrak B,
	\end{align*} for all $\sigma_\mathfrak A\in g\text-Sub(\mathcal{L,L_\mathfrak A})$ and $\sigma_\mathfrak B\in g\text-Sub(\mathcal{L,L_\mathfrak B})$, and every $\mathcal Z$-variable in $\tilde s$ occurs in $\tilde t$ (recall that every $\mathcal Z$-variable in $\tilde t$ occurs in $\tilde s$ by assumption), then
	\begin{align*} 
		(\mathfrak{A,B})\models a:b::c:d.
	\end{align*}
\end{enumerate}
\end{lemma}
\begin{proof} The implication in \prettyref{equ:UL1} guarantees that $\tilde s\mathbf{[x](z,Z)}\to \tilde t\mathbf{[x](z,Z)}$ only justifies the arrow proportion $a\to b\righttherefore c\to d$ and no other arrow proportion $a\to b\righttherefore c\to d'$, for some ground term $d'\in\mathcal{T_{L_\mathfrak A}(\emptyset)}$ with ${d'}^\mathfrak B\neq d^\mathfrak B$, in $(\mathfrak{A,B})$. All the other statements are immediate consequences of the first.
\end{proof}

The next theorem studies functional proportions of the form $a:t(a)::c:t(c)$ for some ``transformation'' $t$ and it is a generalization of the Functional Proportion Theorem of \citeA{Antic22,Antic22-3}.

\begin{theorem}[Bilingual Functional Proportion Theorem]\label{thm:BFPT} Let $\tilde t[\mathbf x](z,\mathbf Z)$ be an abstract $\mathcal L$-hedge containing the single $\mathcal Z$-variable $z$, and let $a\in\mathcal{T_{L_\mathfrak A}(\emptyset)}$ and $c\in\mathcal{T_{L_\mathfrak B}(\emptyset)}$ be ground terms. If
\begin{align*}
	(\tilde t\sigma_{\mathfrak B,z\mapsto c})^\mathfrak B=(\tilde t\sigma'_{\mathfrak B,z\mapsto c})^\mathfrak B,
\end{align*} for all $\sigma_{\mathfrak B,z\mapsto c},\sigma'_{\mathfrak B,z\mapsto c}\in g\text-Sub(\mathcal{L,L_\mathfrak B})$, then the $\mathcal L$-justification
\begin{align*} 
	z\to\tilde t[\mathbf x](z,\mathbf Z)
\end{align*} characteristically justifies
\begin{align*}
    (\mathfrak{A,B})\models a\to\tilde t\sigma_{\mathfrak A,z\mapsto a}\righttherefore c\to\tilde t\sigma_{\mathfrak B,z\mapsto c},
\end{align*} for all $\sigma_{\mathfrak A,z\mapsto a}\in g\text-Sub(\mathcal{L,L_\mathfrak A})$ and $\sigma_{\mathfrak B,z\mapsto c}\in g\text-Sub(\mathcal{L,L_\mathfrak B})$ with $(\sigma_\mathfrak A\upharpoonright\mathcal X)=(\sigma_\mathfrak B\upharpoonright\mathcal X)$.
\end{theorem}
\begin{proof} A direct consequence of the Bilingual Uniqueness \prettyref{lem:BUL}.
\end{proof}

\begin{example} Reconsider the arrow proportion
\begin{align*} 
	\dot S(\dot 0)\to\dot S(\dot S(\dot 0)) \righttherefore \dot 1\;\dot +\;\dot 1\to \dot 1\;\dot +\;\dot 1\;\dot +\;\dot 1
\end{align*} of \prettyref{exa:1_2__2_4}. Let
\begin{align*} 
	\tilde t[\mathbf x](z,\mathbf Z):=\dot h(z,Z)
\end{align*} be an abstract $\mathcal L$-hedge not containing $\mathcal X$-variables. Since
\begin{align*} 
	\sigma_\mathfrak M=\{\dot h/\dot +,z/\dot 1\;\dot +\;\dot 1,Z/\dot 1\} \quad\text{and}\quad \sigma'_\mathfrak M=\{\dot h/\dot +,z/\dot 1,Z/\dot 1\;\dot +\;\dot 1\}
\end{align*} are (modulo semantic equivalence in $\mathfrak M$) the only $(\mathcal{L,L_\mathfrak M})$-substitutions satisfying
\begin{align*} 
	(\dot 1\;\dot +\;\dot 1\;\dot +\;\dot 1)^\mathfrak M=(\tilde t\sigma_\mathfrak M)^\mathfrak M=(\tilde t\sigma'_\mathfrak M)^\mathfrak M,
\end{align*} the Bilingual Functional Proportion \prettyref{thm:BFPT} yields
\begin{align*} 
	(\mathfrak{N,M})\models\dot S(\dot 0)\to\dot h(z,Z)\{\dot h/\dot S,z/\dot 0,Z/\dot\lambda\} \righttherefore \dot 1\;\dot +\;\dot 1\to\dot h(z,Z)\{\dot h/\dot +,z/\dot 1\;\dot +\;\dot 1,Z/\dot 1\},
\end{align*} which is equivalent to
\begin{align*} 
	(\mathfrak{N,M})\models\dot S(\dot 0)\to\dot S(\dot S(\dot 0)) \righttherefore \dot 1\;\dot +\;\dot 1\to\dot 1\;\dot +\;\dot 1\;\dot +\;\dot 1.
\end{align*}
\end{example}

\begin{theorem}\label{thm:p-axioms} The bilingual framework of this paper satisfies the same set of proportional axioms as in \citeA[Theorem 28]{Antic22}.
\end{theorem}
\begin{proof} An immediate consequence of the fact that the unilingual framework is a special case of the bilingual one (\prettyref{con:unilingual}) which means that all counterexamples in the proof of \citeA[Theorem 28]{Antic22} can be transferred.
\end{proof}

\section{Problems}\label{sec:Problems}

This section lists problems which remained unsolved in this paper and appear to be interesting lines of future research \cite<and see>[§8]{Antic22}.

\begin{problem} Generalize the concepts and results of this paper from universal algebra to full first-order logic in the same way as \citeA{Antic22,Antic22-3} has been generalized to \cite{Antic23-4} thus giving raise to bilingual {\em first-order} analogical proportions. For this, it will be necessary to consider formulas over hedges as $\mathcal L$-justifications instead of ordinary formulas.
\end{problem}

\begin{problem} \citeA{Antic23-8} has recently introduced a notion of bilingual homomorphism and isomorphism. Prove generalizations of the First and Second Isomorphism Theorems from \citeA{Antic22,Antic22-3} in the bilingual setting.
\end{problem}

\begin{problem} From a practical point of view, it is important to develop algorithms for the computation of sets of characteristic $\mathcal L$-justifications in abstract and concrete algebras.
\end{problem}


\begin{problem} A {\em proportional functor} \cite{Antic22-4} is a mapping $F:\mathfrak{A\to B}$ satisfying
\begin{align*} 
	(\mathfrak{A,B})\models a:b::\overbrace{Fa^\mathfrak A}^\cdot:\overbrace{Fb^\mathfrak A}^\cdot
\end{align*} for all ground terms $a,b\in\mathcal{T_{L_\mathfrak A}(\emptyset)}$. Proportional functors are mappings preserving the analogical proportion relation and play thus a fundamental role --- therefore study proportional functors in the bilingual setting.
\end{problem}

\begin{problem} Clarify the role of anti-unification in our framework \cite<cf.>{Cerna23}.
\todo[inline]{}
\end{problem}

\section{Conclusion}

The purpose of this paper was to generalize Anti\'c's abstract algebraic framework of analogical proportions from a unilingual to a bilingual setting where the underlying languages of the source and target algebras may differ. This is a major generalization vastly extending the applicability of the underlying framework. In a broader sense, this paper is a further step towards a mathematical theory of analogical reasoning.

\if\isdraft0


\section*{Conflict of interest}

The authors declare that they have no conflict of interest.

\section*{Data availability statement}

The manuscript has no data associated.
\fi

\if\isdraft1\newpage\fi
\bibliographystyle{theapa}
\bibliography{/Users/christianantic/Bibdesk/Bibliography,/Users/christianantic/Bibdesk/Preprints,/Users/christianantic/Bibdesk/Publications,/Users/christianantic/Bibdesk/Unpublished}
\if\isdraft1
\newpage

\section{Bilingual isomorphism theorems}

In this section, we generalize the First and Second Isomorphism Theorems in \citeA{Antic22,Antic22-3} from the unilingual to the bilingual setting of this paper. 

\subsection{Bilingual homomorphisms and isomorphisms}

A mapping $F:\mathfrak{a\to b}$ is a {\em bilingual homomorphism} iff it satisfies
\begin{align*} 
	Ff_i^\mathfrak A(a_1,\ldots,a_{r(f_i)})=g_i^\mathfrak B(Fa_1,\ldots,Fa_{r(g_i)}),\quad\text{if $r(f_i)\geq r_{\mathcal L_\mathfrak A}(g_i)$},
\end{align*} and
\begin{align*} 
	Ff_i^\mathfrak A(a_1,\ldots,a_{r(f_i)})=g_i^\mathfrak B(Fa_1,\ldots,Fa_{r(g_i)},\ldots,Fa_{r(g_i)}),\quad\text{if $r(f_i)<r(g_i)$},
\end{align*} for all $i\in I$ and all elements. Notice that in case $g_i^\mathfrak B=f_i^\mathfrak B$, for all $i\in I$, then the above criterion is the usual one for $F$ to be a unilingual homomorphism. A {\em bilingual isomorphism} is a bijective bilingual homomorphism. We say that $\mathfrak A$ is {\em bilingually isomorphic} to $\mathfrak B$ --- in symbols, $\mathfrak A\lesssim\mathfrak B$ --- iff there is a bilingual isomorphism from $\mathfrak A$ to $\mathfrak B$. Notice that the relation is, in general, not symmetric as justified by the following counterexample. Define the {\em $i$-th projection of the $n$-ary successor function} $S_i^n:\mathbb N^n\to\mathbb N$ by
\begin{align*} 
	S_i^n(a_1,\ldots,a_n):=a_i+1.
\end{align*} Of course, $S_i^n$ does only depend on its $i$-th argument. We denote the unary successor function $S_1^1$ simply by $S$. Let $F:(\mathbb N,S_i^n)\to (\mathbb N,S)$ be a bijective mapping. For $F$ to be a bilingual isomorphism, it needs to satisfy
\begin{align*} 
	FS_i^n(a_1,\ldots,a_n))=SFa_1
\end{align*} or, in other words,
\begin{align*} 
	F(a_i+1)=Fa_1+1,\quad\text{for all $a_i,a_1\in N$},
\end{align*} which is in general impossible for $i\neq 1$. On the other hand, for $i=1$, the identity function is easily seen to be a bilingual isomorphism from $(\mathbb N,S_1^n)$ to $(\mathbb N,S)$, for any $n\geq 1$. Formally, we have
\begin{align*} 
	(\mathbb N,S_1^n)\lesssim (\mathbb N,S) \quad\text{whereas}\quad (\mathbb N,S_i^n)\not\lesssim (\mathbb N,S),\quad\text{for any $2\leq i\leq n$.}
\end{align*} On the other hand, the identity function is a bilingual isomorphism $(\mathbb N,S)\to (\mathbb N,S_i^n)$, for any $1\leq i\leq n$, since we have
\begin{align*} 
	S(a)=S_i^n(a,\ldots,a).
\end{align*} That is,
\begin{align*} 
	(\mathbb N,S)\lesssim (\mathbb N,S_i^n),\quad\text{for all $1\leq i\leq n$}.
\end{align*} In particular, we have
\begin{align*} 
	(\mathbb N,S)\lesssim (\mathbb N,S_2^2) \quad\text{whereas}\quad (\mathbb N,S_2^2)\not\lesssim (\mathbb N,S),
\end{align*} which shows that being bilingually isomorphic is not a symmetric relation.

\begin{example} 
\todo[inline]{}
\end{example}

For any $\mathfrak A$-assignment $\alpha$ and any function $F:a\to b$, we define the {\em induced $\mathfrak B$-assignment} by
\begin{align*} 
	zF\alpha&:=Fz\alpha\\
	ZF\alpha&:=FZ\alpha.
\end{align*} We say that $F$ {\em respects} a hedge $\tilde s$ iff for each $\mathfrak A$-assignment $\alpha$,
\begin{align*} 
	F(\tilde s\sigma)^\mathfrak A=(\tilde sF\alpha)^\mathfrak B.
\end{align*}

\begin{lemma} Bilingual isomorphisms respect hedges.
\end{lemma}
\begin{proof} By structural induction on the shape of $\tilde s$:
\begin{enumerate}
	\item If $\tilde s$ is...
\end{enumerate}
\todo[inline]{}
\end{proof}

\subsection{First and second bilingual isomorphism theorem}

A {\em proportional functor} \cite{Antic22-4} is a mapping $F:\mathfrak{a\to b}$ satisfying
\begin{align*} 
	(\mathfrak{A,B})\models a:b::Fa:Fb,
\end{align*} for all $a,b\in A$.

We are now ready to state and prove that bilingual analogical proportions are compatible with bilingual isomorphisms as desired.

\begin{lemma}[Bilingual Isomorphism Lemma]\label{lem:BIL} For any bilingual isomorphism $F:\mathfrak{a\to b}$ and any elements $a,b\in A$, we have
\begin{align}\label{equ:BIL} 
	Jus_\mathfrak A(a\to b)\subseteq Jus_\mathfrak B(F(a)\to F(b)).
\end{align} 
\end{lemma}
\begin{proof} 
\todo[inline]{$F^{-1}$ is not necessarily a bilingual iso!}
\todo[inline]{Is there an example where $Jus_\mathfrak A(a\to b)\subsetneq Jus_\mathfrak B(F(a)\to F(b))$?}
\end{proof}

\begin{remark} Since $F^{-1}$ is in general not a bilingual isomorphism, we cannot expect the set inclusion in \prettyref{equ:BIL} to be an identity. This is different from the unilingual framework \cite<see>[Isomorphism Lemma]{Antic22}.
\end{remark}

\begin{theorem}[First Bilingual Isomorphism Theorem] Every bilingual isomorphism is a proportional functor.
\end{theorem}
\begin{proof} We first want to prove
\begin{align*} 
	(\mathfrak{A,B})\models a\to b \righttherefore Fa\to Fb.
\end{align*} If
\begin{align*} 
	Jus_\mathfrak A(a\to b)\cup Jus_\mathfrak B(Fa\to Fb)=\emptyset,
\end{align*} we are done. Otherwise, there is at least one non-trivial justification $\tilde s\to\tilde t$ such that
\begin{align*} 
	\tilde s\to\tilde t\in Jus_\mathfrak A(a\to b) \quad\text{or}\quad \tilde s\to\tilde t\in Jus_\mathfrak B(Fa\to Fb).
\end{align*} If $\tilde s\to\tilde t\in Jus_\mathfrak A(a\to b)$, then the Bilingual Isomorphism \prettyref{lem:BIL} implies $\tilde s\to\tilde t\in Jus_\mathfrak B(Fa\to Fb)$, which further implies
\begin{align*} 
	Jus_{(\mathfrak{A,B})}(a\to b\righttherefore c\to d)\neq\emptyset.
\end{align*} Otherwise, if $\tilde s\to\tilde t\in Jus_\mathfrak B(Fa\to Fb)$...
\todo[inline]{We cannot use the BIL in the same way as before}
\end{proof}

The following counterexample...

\begin{example} 
\todo[inline]{}
\end{example}

\begin{theorem}[Second Bilingual Isomorphism Theorem]
\end{theorem}
\begin{proof} 
\todo[inline]{}
\end{proof}



\fi
\end{document}